\documentclass[11pt,a4paper]{article}
\usepackage{amsfonts,amssymb,amsmath,amsthm,cite}
\usepackage{graphicx}

\usepackage[applemac]{inputenc}

\evensidemargin=\oddsidemargin
\DeclareMathOperator{\Dom}{Dom}      
\DeclareMathOperator{\End}{End}      
\DeclareMathOperator{\Diff}{Diff}           
\DeclareMathOperator{\Id}{Id}                 
\DeclareMathOperator{\Res}{Res}         
\DeclareMathOperator{\Tad}{Tad}          
\DeclareMathOperator{\Tr}{Tr}                 

\newtheorem{assumption}{Assumption}[section]
\newtheorem{theorem}[assumption]{Theorem}

\newtheorem{lemma}[assumption]{Lemma}
\newtheorem{definition}[assumption]{Definition}
\newtheorem{prop}[assumption]{Proposition}
\newtheorem{remark}[assumption]{Remark}

\renewcommand{\th}{\theta}
\newcommand{\A}{\mathcal{A}}              
\renewcommand{\a}{\alpha}                    
\newcommand{\B}{\mathcal{B}}              
\newcommand{\C}{\mathbb{C}}              
\newcommand{\del}{\partial}                    
\newcommand{\DD}{\mathcal{D}}           
\newcommand{\eps}{\varepsilon}            
\newcommand{\F}{\mathcal{F}}                
\newcommand{\ga}{\gamma}                   
\renewcommand{\H}{\mathcal{H}}           
\newcommand{\half}{{\mathchoice{\thalf}{\thalf}{\shalf}{\shalf}}}
\newcommand{\la}{\lambda}                   
\newcommand{\<}{\langle}       
     
\newcommand{\N}{\mathbb{N}}             
\newcommand{\om}{\omega}                 
\newcommand{\ol}{\overline}                  
\newcommand{\OO}{\mathcal{O}}          

\newcommand{\R}{\mathbb{R}}               

\newcommand{\set}[1]{\{\,#1\,\}}               
\newcommand{\shalf}{{\scriptstyle\frac{1}{2}}} 
\renewcommand{\SS}{\mathcal{S}}        
\newcommand{\thalf}{\tfrac{1}{2}}            
\newcommand{\Afr}{\mathfrak{A}}           
\newcommand{\wh}{\widehat}                  
\newcommand{\wt}{\widetilde}                 
\newcommand{\Z}{\mathbb{Z}}                 

\def\<#1,#2>{\langle#1\,,\,#2\rangle}      

\newcommand{\be}{\begin{enumerate}}
\newcommand{\ee}{\end{enumerate}}

\newbox\ncintdbox \newbox\ncinttbox
\setbox0=\hbox{$-$} \setbox2=\hbox{$\displaystyle\int$}
\setbox\ncintdbox=\hbox{\rlap{\hbox
    to \wd2{\kern-.1em\box2\relax\hfil}}\box0\kern.1em}
\setbox0=\hbox{$\vcenter{\hrule width 4pt}$}
\setbox2=\hbox{$\textstyle\int$}
\setbox\ncinttbox=\hbox{\rlap{\hbox
    to \wd2{\kern-.14em\box2\relax\hfil}}\box0\kern.1em}
\newcommand{\ncint}{\mathop{\mathchoice{\copy\ncintdbox}
    {\copy\ncinttbox}{\copy\ncinttbox}
    {\copy\ncinttbox}}\nolimits}
\newcommand{\sg}{\sigma}                              

\newcommand{\twobytwo}[4]{\begin{pmatrix}#1 & #2 \\ #3 & #4
                            \end{pmatrix}} 

\begin{document}

\thispagestyle{empty}

\begin{center} 

CENTRE DE PHYSIQUE TH\'EORIQUE$\,^1$\\
CNRS--Luminy, Case 907\\
13288 Marseille Cedex 9\\
FRANCE\\

\vspace{3cm}

{\Large\textbf{Spectral triples and manifolds with boundary}} \\
\vspace{0.5cm}

{\large  B. Iochum$^{1, 2}$, C. Levy$^{3}$} \\

\vspace{1.5cm}

{\large\textbf{Abstract}}
\end{center}
We investigate manifolds with boundary in noncommutative geometry. Spectral triples associated to a 
symmetric differential operator and a local boundary condition are constructed. We show that there is no tadpole for classical Dirac operators with a chiral boundary condition on certain manifolds.
\begin{quote}

\end{quote}

\vspace{2cm}



CPT-P003-2010

\vspace{4cm}

{\small
\noindent $^1$ UMR 6207

-- Unit\'e Mixte de Recherche du CNRS et des
Universit\'es Aix-Marseille I, Aix-Marseille II et de l'Universit\'e
du Sud Toulon-Var (Aix-Marseille Université)

-- Laboratoire affili\'e \`a la FRUMAM -- FR 2291\\
$^2$ Also at Universit\'e de Provence,
iochum@cpt.univ-mrs.fr\\
$^3$ University of Copenhagen, Department of mathematical sciences, Universitetsparken 5, DK-2100 
Copenhagen, Denmark, levy@math.ku.dk
}

\newpage


\section{Introduction}

Noncommutative geometry has, in its numerous motivations, the conceptual 
understanding of different aspects of physics \cite{Book,ConnesMarcolli}. In particular, the spectral 
approach which is deeply encoded  
in the notion of a spectral triple, is not only motivated by the algebra of quantum observables $\A$ acting 
on the Hilbert space $\H$ of physical states, but also by classical physics like general relativity. For 
instance, a Riemannian compact spin manifold can be reconstructed only via properties of a commutative 
spectral triple $(\A,\H,\DD)$ \cite{CReconstruction} where the last piece $\DD$ is a selfadjoint operator 
acting on $\H$ playing the role of a Dirac operator which can fluctuates: $\DD$ is then replaced by 
$\DD_A:=\DD+A$ where $A$ is a selfadjoint one-form.

The spectral action $\SS$ of Chamseddine--Connes \cite{CC} associated to a triple $(\A,\H,\DD)$ is the 
trace of $\Phi(\DD_A^2/\Lambda^2)$ where $\Phi$ is a positive function and $\Lambda$ plays the role of a 
cut-off. This can be written (under some conditions on the spectrum) as a series of noncommutative 
integrals 
\begin{align}
   \label{formuleaction}
	\SS(\DD_{A},\Phi,\Lambda) \, = \,\sum_{k\in Sd^+} \Phi_{k}\,
	\Lambda^{k} \ncint \vert \DD_{A}\vert^{-k} + \Phi(0) \,
	\zeta_{\DD_{A}}(0) +\mathcal{O}_{\Lambda\to \infty}(\Lambda^{-1})
\end{align}
where $\Phi_{k}=\half\int_{0}^{\infty} \Phi(t) \, t^{k/2-1} \, dt$, $Sd^+$ is the strictly positive part of the 
dimension spectrum of the spectral triple and the noncommutative integral $\ncint X$ for $X$ in the 
algebra $\Psi(\A)$ of pseudodifferential operators, is defined by 
$\ncint X:=\underset{s=0}{\Res} \,\Tr \big(X \vert \DD \vert^{-s}\big)$. 

Since $\ncint$ is a trace on $\Psi(\A)$ (non necessarily positive), it coincides (up to a scalar) 
with the Wodzicki residue \cite{Wodzicki1,Wodzicki3} in the case of a commutative geometry 
where $\A$ is the algebra of $C^\infty$ functions on a manifold $M$ without boundary:
in a chosen coordinate system and local trivialization $(x,\xi)$ of $T^*M$, this residue is
\begin{align*}
Wres(X):=\int_M \,\int_{S_x^*M} \Tr\big(\sigma_{-d}^X\,(x,\xi)\big)\,b\vert d\xi \vert\,\vert dx\vert,
\end{align*}
where $\sigma_{-d}^X$ is the symbol of the classical pseudodifferential operator $X$ which is 
homogeneous of degree $-d:=-\text{dim}(M)$, $d\xi$ is the normalized restriction of the volume form to the 
unit sphere $S_x^*M \simeq \mathbb{S}^{d-1}$. 
The Dixmier's trace $\Tr_\omega$ \cite{Dix} concerns compact 
operators $X$ with singular values $\set{\mu_k}$ satisfying $\sup_{N\rightarrow \infty} \, a_N < \infty$ 
where $a_N=\log(N)^{-1}\sum_{k=1}^N \mu_k$ and $\Tr_{\omega}(X)$ is defined after a choice of an 
averaging  procedure $\omega$ such that $\Tr_{\omega}(X)=\lim_N \, a_N$ when $a_N$ converges.
As shown in \cite{Connesaction}, $\ncint$ coincides (still up to a universal scalar) with $\Tr_\omega$ 
when $X$ has order $-d$.

When $M$ has a boundary, the choice of an appropriate differential calculus is delicate. In the 
noncommutative framework, the links between Boutet de Monvel's algebra, Wodzicki's residue, 
Dixmier's trace or Kontsevich--Vishik's trace \cite{KV} have been clarified \cite{FGLS, Schrohe, GSc, NS, 
ANS, GSc1} including the case of log-polyhomogeneous symbols \cite{Lesch1}.

From a physics point of view, applications of noncommutative integrals on manifolds to classical gravity 
has begun with Connes' remark that $\ncint \DD^{-2}$ coincides in dimension 4 with Einstein--Hilbert 
action, a fact recovered in \cite{Kastler, KW}. Then, a generalization to manifolds with boundaries was 
proposed in \cite{Wang1,Wang2,Wang3,Uga}. From the quantum side, a noncommutative approach of the  
unit disk is proposed in \cite{CKW}.

\vspace{0.3cm}
Nevertheless, a construction of spectral triples in presence of boundary is not an easy task, although a 
general approach of boundary spectral triples has been announced in \cite{Connesbord}. 
 
First examples appear with isolated conical singularity in \cite{Lescure}, a work related to some extend to 
\cite{Schrohe1,Lesch} when the spectrum dimension is computed. The difficulty is to find an 
appropriate boundary condition which preserves not only the selfadjointness of the realization of $\DD$ 
but also the ellipticity. A special case of boundary also appears in non-compact manifolds when one 
restricts the operator $\DD$ to a bounded closed region: for instance, this trick was used in 
\cite{Rennie,Yang}. 

The choice of a chiral boundary condition, already considered in \cite{BG2} for mathematical reasons, is 
preferred in \cite{CC2} for physical reasons: firstly, it is consistent with a selfadjoint and elliptic 
realization, and secondly, it is a local boundary condition contrary to the standard 
APS' one which is global \cite{APS}. Thirdly, it gives a similar ratio and signs for the second term of the 
spectral action (the first one being the cosmological constant), namely the scalar curvature of the manifold 
and the extrinsic curvature of the boundary, as in the Euclidean action used in gravitation \cite{HH}. Since 
there are a lot of possible choices, this last consideration deserves attention.

\vspace{0.3cm}
Here, we first show a construction for manifolds with boundary that actually produces a spectral 
triple, and then give conditions on the algebra of functions on that manifold to get a regular triple (remark 
that the spectral action has only been computed, until now, for spectral triples which are regular). 
While in field theory, the one-loop calculation divergences, anomalies and different asymptotics of the 
effective action are directly obtained from the heat kernel method \cite{Gilkey,Gilkey2,Vassi}, we try to 
avoid this perturbation approach already used in \cite{Tadpole} to prove that there are no tadpoles when 
a reality operator $J$ exists. 
Tadpoles are the $A$-linear terms in \eqref{formuleaction}, like for example $\ncint A\DD^{-1}$. In quantum field
theory, $\DD^{-1}$ is the Feynman propagator and $A\DD^{-1}$ is a
one-loop graph with fermionic internal line and only one external
bosonic line $A$ looking like a tadpole. 


In section \ref{Regularity}, we derive a technical result on regularity of spectral triples which is sufficient  to 
avoid the use of Sobolev spaces of negative order. Then, we recall few basics on the realization of 
boundary pseudodifferential operators and their stability by powers using the Grubb's approach 
\cite{Grubb}. In section \ref{SpectralTriple}, we define an algebra $\A_{P_T}$ compatible with the 
realization of an elliptic pseudodifferential boundary system $\set{P,T}$. A condition on $P$ is given 
which guarantees the regularity of the associated spectral triple. The motivating example of a classical 
Dirac operator is considered in section \ref{CasDirac}. Moreover, the construction of a spectral triple on 
the boundary is revisited in section \ref{tripletbord}. Section \ref{Reality} is devoted to a reality operator 
$J$ on a spectral triple with boundary and some consequences on the tadpoles like $\ncint A \DD^{-1}$ 
which can appear in spectral action.

\section{Regularity}
\label{Regularity}

Let $\N$ be the non-negative integers and $\B(\H)$ be the set of bounded operator on a separable Hilbert space 
$\H$. 

We shall use the following definition of a spectral triple:

\begin{definition}
A spectral triple of dimension $d$ is a triple $(\A,\H,\DD)$ such that $\H$ is a Hilbert space and

\quad - $\A$ is an involutive unital algebra faithfully represented in $B(\H)$,

\quad - $\DD$ is a selfadjoint operator on $\H$ with compact resolvent and its singular values 
$(\mu_n(|\DD|))_n$ are $\OO(n^{1/d})$,

\quad - for any $a\in \A$, $a\Dom \DD \subseteq \Dom \DD$ and the commutator $[\DD,a]$ 
(with domain $\Dom \DD$) as an extension in $B(\H)$ denoted $da$.
\end{definition}

Note that $\Dom |\DD| =\Dom \DD$. We set $\delta(T):={ \overline{[|\DD|,T]}}$, where $\overline{A}$ is the closure 
of the operator $A$, with domain 
\begin{align*}
\Dom
\delta:=\set{T\in \B(\H) \, : \  &T\Dom \DD \subseteq \Dom \DD \text{ and }
[|\DD|,T] \text{ has closure in } \B(\H)}.
\end{align*}

\begin{definition}
A spectral triple $(\A,\H,\DD)$ is said to be regular if $\A$ and $d\A$ are included in 
$\cap_{n\in \N} \Dom \delta^n$.
\end{definition}
As seen in next lemma, it is quite convenient to introduce the following

\begin{definition}
\label{Def}
Given a selfadjoint operator $P$ on a Hilbert space $\H$, let 
$$
\H_P^\infty:=\cap_{k\geq 1} \Dom |P|^k=\cap_{k\geq 1} \Dom P^k.
$$
A linear map from $\H_P^\infty$ into itself which is continuous for the topologies induced by $\H$ is said to be 
$\H_P^\infty$-bounded. 
\end{definition}

Note that $\H_P^\infty$ is a core for any power of $P$ or $|P|$. In particular it is a dense subset of 
$\H$ and $|P|$ is the closure of the essentially selfadjoint operator $|P|_{\vert\H^\infty_P}$.
Note also that for any $k\in \N$, $(1+P^2)^{k/2}$ is a bijection from $\Dom P^k$ onto $\H$, and 
thus,  $(1+P^2)^{-k/2}$ is a bijection from $\H$ onto $\Dom P^k$. As a consequence, 
for any $p\in \Z$, the operators $(1+P^2)^{p/2}$ send bijectively $\H_P^\infty$ onto itself, and
for any $k\in \N$, $|P|^k$ send $\H_P^\infty$ into itself.

Given a selfadjoint $P$, let $\delta',\, \delta_1$ be defined on operators by
$$
\delta'(T):=[|P|,T], \qquad \delta_1(T):=[P^2,T]\,(1+P^2)^{-1/2} \, 
$$
with domains 
$$
\Dom \delta' \text{ (resp.}\Dom \delta_1) :=\set{T \in \H_P^\infty{\text{-bounded operators}\, : \, 
\delta'(T) \,\big(\text{resp. }\delta_1(T)\big) \text{ is } \H_P^\infty \text{-bounded}}}.
$$ 
We record the following lemma, proven by A. Connes: 

\begin{lemma}
\label{RegularityLemma}  \cite[Lemma 13.1 and 13.2]{CReconstruction}.

(i) If $T\in \Dom \delta'$, then the bounded closure $\ol T$ of $T$ is in $\Dom \delta$ and 
$\delta(\ol T) = \ol{\delta'(T)}$. 

So, by induction, if $T\in  \cap_{n\in\N} \Dom \delta'\,^n$ then 
$\ol T \in \cap_{n\in\N} \Dom \delta^n$. 

(ii) Let $T$ be a $\H_P^\infty$-bounded operator.  
If $T\in \cap_{n\in \N}\Dom {\delta_1}^n$, then $T\in  \cap_{n\in\N} \Dom \delta'\,^n$. 

In particular the bounded closure $\ol T$ of $T$ belongs to $\cap_{n\in\N} \Dom \delta^n$.
\end{lemma}

\begin{definition}
Given a Hilbert space $\H$ and a selfadjoint (possibly unbounded) operator $P$, we call Sobolev 
scale on $\N$, a family $(H^k)_{k\in \N}$ of Hilbert spaces such that 

\quad - $H^0=\H$,

\quad- $H^{k+1}$ is continuously included in $H^{k}$ for any $k\in \N$,

\quad - for any $k\in \N$, $\Dom P^k$ is a closed subset of $H^k$.
\end{definition}
    
By closed graph theorem, the last point implies that $P^k$ is continuous from $\Dom P^k$ endowed with 
the $H^k$-topology into $\H$.

Similar abstract Sobolev scales, defined as domains of the powers of an abstract differential operator, have been 
considered in \cite{Higson2}. A corresponding criterion for regularity has been obtained in \cite[4.26 Theorem]
{Higson2}. The scale we consider here will correspond in section \ref{SpectralTriple} to the Sobolev spaces (on the 
manifold with boundary) and not to the domains of the powers of the realization of a first order pseudodifferential 
operator. In the case without boundary, these scales coincide.   

When $T$ is a $\H_P^\infty$-bounded operator, we shall denote $T^{(k)}:=[P^2,\cdot]^k(T)$ for any 
$k\in \N$.
\begin{lemma}
\label{RegularityLem2} 
Let $P$ be a selfadjoint operator and $T$ be a $\H_P^\infty$-bounded operator.
Suppose that there is a Sobolev scale $(H^k)_{k\in \N}$ such that for any $k\in \N$, 
$T^{(k)}$ is continuous from $\H^\infty_P$ with the $H^k$-topology into $\H^\infty_P$ with the 
$\H$-topology. Then $\ol T\in \cap_{n\in\N} \Dom \delta^n$.
\end{lemma}

\begin{proof} 
The operator $(P-i)^k = \sum_{j=0}^k \tbinom{k}{j}(-i)^{k-j}P^j$ is
continuous from $\Dom P^k$ with the $H^k$-topology into $\H$. Since $P$ is
selfadjoint, $P-i$ is a bijection from $\Dom P$ onto $\H$, and by composition, $(P-i)^{k}$ is a
bijective map from $\Dom P^k$ onto $\H$. The inverse mapping theorem now
implies that $(P-i)^{-k}$ is continuous from $\H$ onto $\Dom P^k$ with the
$H^k$-topology. Moreover, $(1+P^2)^{-k/2} = (P-i)^{-k} B$
where $B:= (P-i)^{k}(1+P^2)^{-k/2}$ is, by spectral theory, a bijective operator in $\B(\H)$. As a 
consequence, $(1+P^2)^{-k/2}$ is continuous from $\H$ onto $\Dom P^k$ with the $H^k$-topology. 
In particular,$(1+P^2)^{-k/2}$ is continuous from $\H^\infty_P$ with the $\H$-topology, into 
$\H^\infty_P$ with the $H^k$-topology. 

So the hypothesis gives that 
$T^{(k)}(1+P^2)^{-k/2}=\big([P^2,\cdot] \, (1+P^2)^{-1/2}\big)^k(T) =\delta_1^k(T)$ 
is a $\H_P^\infty$-bounded operator. The result follows from Lemma \ref{RegularityLemma}.
\end{proof}

The previous lemma essentially implies that, in order to prove the regularity of a spectral triple 
$(\A,\H,\DD)$, it is sufficient to construct a Sobolev scale $(H^k)_{k\in \N}$ adapted to $\H$ and $\DD$, 
such that the operators $\DD^k$ and $T^{(k)}$ behave respectively as operators of ``order'' $k$  with
respect to the Sobolev scale, when $T$ is any element of $\A \cup d\A$.

\begin{remark}
By Lemma \ref{RegularityLemma}, it is possible to obtain regularity without using Sobolev spaces of 
negative order, implicitly used for instance in \cite[Theorem 11.1]{Polaris} which follows the original 
argument of \cite{Cgeom}. This shall considerably simplify the proof of the regularity of the spectral triple 
on manifolds with boundary, since the continuity of realizations of elliptic boundary differential operators is 
usually established on Sobolev spaces of positive order \cite{Grubb, BL,BLZ,BW}. Note however that it 
may be possible to deal with negative orders by using the technique of transposition described in 
\cite{LM}.

Another approach of regularity can be found in \cite{Higson1,Higson2,Otogo}.
\end{remark}

\section{Background on elliptic systems on manifolds with boundary}

We review in this section a few definitions and basic properties about Sobolev spaces in manifolds with 
and without boundary and boundary pseudodifferential operators choosing Boutet de Monvel's calculus. 
More details and proofs can be found in classical references like \cite{Hormander,Grubb}.

Let $\wt M$ be a smooth compact manifold without boundary of dimension $d$ and $\wt E$ be a smooth 
hermitian vector bundle on $\wt M$. Let $M$ be an open submanifold of $\wt M$ of dimension $d$ such that 
$\ol M$ (topological closure) is a compact manifold with nonempty boundary $N:=\partial \ol M=\ol M \backslash M$. As 
a consequence, $N$ is a smooth compact submanifold of $\wt M$ without boundary of dimension $d-1$.

The sub-bundle of $\wt E$ on $\ol M$ (resp. $N$) is denoted $E$ (resp. $E_N$). 
We denote $H^s(\wt E)$, $H^s(E)$ the Sobolev spaces of order $s\in \R$ respectively on 
$\wt M$ with bundle $\wt E$ and $\ol M$ with bundle $E$. Recall that by definition 
$$
H^s(E) := H^s(\ol M,E):= r^+ \big(H^s(\wt E)\big)\, 
$$
where $r^+$ is the restriction to $M$. We refer to \cite[p. 496]{Grubb} for the definition of the topology of $H^s(E)$.

Remark that, for a given manifold $\ol M$ with boundary $\partial \ol M$, it is always possible to construct 
$\wt M$ with previous properties. Moreover, there exist constructions of invertible double for Dirac 
operators and more general first order elliptic operators on closed double of $\ol M$\cite{BW,BBL,BLZ}.

We denote $\Psi^k(\wt E)$ (resp. $\Diff^k(\wt E)$) the space of pseudodifferential (resp. differential) 
operators of order $k$ on $(\wt M,\wt E)$. Any element of $\Psi^k(\wt E)$ is a linear continuous
operator from $H^s(\wt E)$ into $H^{s-k}(\wt E)$, for any $s\in \R$.
  
We set
$$
C^\infty(\ol M, E):= r^+\,\big( C^\infty(\wt M,\wt E) \big) \, , \qquad
C^\infty(\ol M):=r^+ \big(C^\infty(\wt M)\big) \, .
$$
Despite notations, note that the spaces $H^s(E)$, $C^\infty(\ol M,E)$ and $C^\infty(\ol M)$ are spaces of 
functions defined on $M$, and not on $\ol M$, as $r^+$ is the restriction on $M$. Remark that $C^\infty(\ol M)$ 
is identified to $C^\infty(\ol M) \Id_{E_{|M}}$, so that $C^\infty(\ol M)$ can be seen as an algebra of bounded 
operators on $H^s(E)$, and in particular on $H^0(E)=L^2(\ol M, E)= r^+ L^2(\wt M,\wt E)$.

A differential operator $P$ on $\ol M$ is by definition a differential operator on $M$ with coefficients in 
$C^\infty(\ol M, E)$. We denote $\Diff^k(\ol M,E)$ the space of differential operators of order $k\in \N$ 
on $\ol M$.

Any element of $\Diff^k(\ol M, E)$ can be extended uniquely as a linear continuous operator from 
$H^s(E)$ into $H^{s-k}( E)$, for any $s\in \R$.
 
Finally, note that $C^\infty(\wt M,\wt E)$ is dense in any $H^p(\wt E)$ and $C^\infty(\ol M, E)$) is also dense 
in any $H^p(E):=r^+H^p(\wt E)$, and moreover,
$$
\bigcap_{p\in \N} H^p(\wt E)=C^\infty(\wt M,\wt E), \qquad \bigcap_{p\in \N} H^p(E)=C^\infty(\ol M,E)\,.
$$

The extension by zero operator $e^+$ is a linear continuous operator from $H^s(E)$ into $H^s(\wt E)$ for 
any $s\in ]-\half,\half[$ such that $e^+(u)=u$ on $M$ and $e^+(u)(x)=0$ for any 
$u \in C^\infty(\ol M,  E)$ and $x\in \wt M\backslash M$. 

For any $P\in \Psi^k(\wt E)$, we define its truncation to $M$ by
$$ 
P_+ := r^+ \, P \, e^+\, .
$$

Recall that $P\in \Psi^k(\wt E)$, $k\in \Z$, is said to satisfy the transmission condition if $P_+$ maps 
$C^\infty(\ol M,E)$ into itself which means that "$P_+$ preserves $C^\infty$ up to the boundary". In 
particular, any differential operator satisfies this condition. 

It turns out that if $P$ satisfies the transmission condition, $P_+$ can be seen as linear continuous 
operator from $H^s(E)$ into $H^{s-k}(E)$ for any $s>-\half$ 
(\cite[2.5.8 Theorem, 2.5.12 Corollary]{Grubb}).

We refer to \cite{Grubb, Grubb1} for all definitions of elliptic boundary system, normal trace operator and 
singular Green operator but to fix notations, we recall the Green formula 
\cite[1.3.2 Proposition]{Grubb}: for $P\in \Psi^k(\wt E)$, $k\in \N$ and for any $u,v\in C^\infty(\ol M,E)$,  
\begin{align}
\label{green}
(P_+ u ,v )_M - (u,(P^*)_+ v)_M = (\Afr_P\, \rho u , \rho v)_N
\end{align}
where $(u,v)_X:=\int_X u(x) \overline{v(x)}\,dx$ (if defined) and $\rho=\set{\ga_0,\cdots,\ga_{k-1}}$ is the 
Cauchy boundary operator given by $\ga_j u=(-i\partial_d)^j\, u_{\vert N}$ ($\partial_d$ being the interior 
normal derivative) and $\Afr_P$ is the Green matrix associated to $P$.

Here, $x=(x',x_d)$ is an element of $\ol M=N\sqcup M$ with $(x',0)\in N$ and $x_d$ denotes a normal 
coordinate. By \cite[Lemma 1.3.1]{Grubb}, $P=A+P'$ with $A=\sum_{l=0}^k S_l\,(-i\partial_d)^l$, where 
$S_l$ is a tangential differential operator of order $k-l$ supported near $N=\partial M$, and $P'$ is a 
pseudodifferential operator of order $k$ satisfying $(P'_+ u ,v )_M - (u,(P'^*)_+ v)_M = 0$. 
The Green matrix satisfies
$$ 
\Afr_P =(\Afr_{j,l})_{j,l=0,\cdots,k-1} \text{ with }
\Afr_{j,l}(x',D'):=iS_{j+l+1}(x',0,D') + \text{ lower-order terms} \, .
$$
and $\Afr_{jl}$ is zero if $j+l+1>k$. 

\begin{remark}
\label{casDirac}
When $P$ is a pseudodifferential operator of order 1, $\Afr_P$ is an endomorphism on the boundary 
$N$. For instance, if $P$ is a classical Dirac operator acting on a Dirac bundle, then $\Afr_P=-i \ga_d$ 
where $\ga$ is the (selfadjoint) Clifford action and \{$e_i$\}, $i=1,\ldots,d$ is a (local) orthonormal frame of $TM$ 
such that $e_d$ is the inward pointing unit normal, and $\ga_i:=\ga(e_i)$. This $\Afr_P$ corresponds to $-J_0$ in 
\cite{BBL,BLZ}.
\end{remark}

When $P\in \Psi^n(\wt E)$, $n\in \N$,  satisfies the transmission condition, $G$ is a singular Green 
operator of order $n$ and class $\leq n$ and $T:=\set{T^0,\cdots,T^{n-1}}$ is a system of normal trace 
operators associated with the order $n$, each $T^i$ going from $E$ to $E_N$, then the \textit{($H^n$-)
realization of the system $\set{P_+ +G, \,T}$} is the operator $(P+G)_T$ defined as the operator acting like 
$P_+ +G$ with domain 
$$
\Dom \,(P+ G)_T := \set{\psi \in H^n(\ol M,E) \ : \ T \psi = 0}.
$$
These realizations are always densely defined, and since $T$ is continuous from $H^n(E)$ into 
$\prod_{j=0}^{n-1} H^{n-j-\half}(E_N)$, $\Dom\, (P+ G)_T$ is a closed subset of $H^n(E)$. Recall that $P_+ + G$ is 
continuous from $H^{s}(E)$ into $H^{s-n}(E)$ for any $s>n-\half$.

We shall assume that all pseudodifferential operators satisfy the transmission property at the boundary.

We record here the following proposition, which is a direct application of
\cite[1.4.6 Theorem, Corollary 2.5.12, 2.7.8 Corollary]{Grubb}: 

\begin{prop}
\label{propcomposition}
Let $\set{P_+ +G, \,T}$ be an elliptic system of order $n$ ($G$ being with class
$\leq n$) with $T$ a system of normal trace operators associated with the order $n$.
 
Then for any $k\in \N$, there exist a singular Green operators $G_k$ of class $\leq nk$, and
a system of normal trace operators $T_k$ associated with the order $nk$ such that $\big((P+G)_T\big)^k$ 
is the realization of the elliptic system $\set{(P^k)_+ + G_k, \, T_k}$ of order $nk$. 

Moreover, $\Dom {((P+G)_T)^k}$ is a closed subset of $H^{nk}(E)$, and $(P^k)_+ + G_k$ is continuous from 
$H^s(E)$ into $H^{s-nk}(E)$ for $s>nk-\half$.
\end{prop}

When $P,P'$ are in $\Psi^\infty(\wt E)$, the leftover of $P$ and $P'$
$$
L(P,P'):= (P P')_+ -P_+ P'_+
$$
is a singular Green operator of order $k+k'$ and class $\leq k+k$' when the order of $P$ and $P'$ are $k$ 
and $k'$.

The following result is a consequence of \cite[(2.6.27)]{Grubb}, but we give a short proof.

\begin{lemma} 
\label{leftoverlem}
Let $P,P'\in \Psi^\infty(\wt E)$.

(i) If $P$ is differential, then $L(P,P')=0$.

(ii) If $P'$ is an endomorphism on $(\wt M, \wt E)$ (differential operator of order 0), then $L(P,P')=0$.
\end{lemma}
\begin{proof}
$(i)$ From the locality of differential operators, we see that $r^+ P e^+ r^+ = r^+ P$. It follows that 
$L(P,P')=0$.

$(ii)$ Since $e^+ r^+P'e^+ =P' e^+$, the result follows. 
\end{proof}

\section{Spectral triples on manifolds with boundary}
\label{SpectralTriple}

\subsection{Spectral triples on $\wt M$ and $\ol M$}

Since it is a first step to the main theorem of this section, we record the following known fact in 
noncommutative geometry \cite{CM}: any elliptic pseudodifferential operator of first order on compact 
manifolds, whose square has a scalar principal symbol, yields a regular spectral triple with the algebra 
of smooth functions. 

Recall that the principal symbol $\sigma_d(P)$ of $P\in \Psi^d(\wt E)$ is said to be scalar when it is of 
the form $\sg \Id_E$ with $\sg\in C^\infty(T^*M,\C)$.

\begin{prop} 
\label{propmtilde}
Let $P\in \Psi^1(\wt E)$ be an elliptic symmetric pseudodifferential operator of order one on $\wt M$ such 
that the principal symbol of $P^2$ is scalar. Then $\big(C^\infty(\wt M), L^2(\wt E), P\big)$ is a regular  
spectral triple of dimension $d$.
\end{prop}

\begin{proof} 
Since $P\in \Psi^1(\wt M)$ is an elliptic symmetric operator on $L^2(\wt E)$, it is 
selfadjoint with domain $H^1(\wt E)$. Any $a\in C^\infty(\wt M)$ is represented by the left multiplication 
operator on $L^2(\wt E)$ which is bounded. Since $a$ is scalar, the commutator $[P,a]$ is a 
pseudodifferential operator of order 0 and thus can be extended as a bounded operator on $L^2(\wt E)$. 

Ellipticity implies that $\Dom P^k = H^k(\wt E)$ for any $k\in \N$, where here $P^k$ is the
composition of operators $P:H^1(\wt E)\to L^2(\wt E)$ as unbounded operators in $L^2(\wt E)$. The 
operator $(P-i)^d$ is thus continuous and bijective from $H^d(\wt E)$ onto $\H$, and by inverse mapping 
theorem $(P-i)^{-d}$ is continuous from $\H$ onto $H^d(\wt E)$. By a classical result (see for instance 
\cite[A.4 Lemma]{Grubb}), $(P-i)^{-d}$ is compact, and the dimension of the given triple is $d$.

It remains to check the regularity. By Sobolev lemma, $\H_P^\infty=C^\infty(\wt M,\wt E)$ so that, for any 
$a\in C^\infty(\wt M)$, $a_{|\H_P^\infty}$ and $da_{|\H_P^\infty}$,  are $\H_P^\infty$-bounded 
operators. Since the principal symbol of $P^2$ is scalar, $(a_{|\H_P^\infty})^{(k)}$ and 
$(da_{|\H_P^\infty})^{(k)}$ are pseudodifferential operators of order $k$ defined on $\H_P^\infty$. 
Applying now Lemma \ref{RegularityLem2} with the Sobolev scale $\big(H^k(\wt E)\big)_{k\in \N}$ 
yields the result.
\end{proof}

In the case of manifolds with boundary, the full algebra $C^\infty(\ol M)$ cannot yield, in general, regular 
spectral triples on $\ol M$, because there is a conflict between the necessity of selfadjointness for the 
realization $P_T$ which is implemented by a boundary condition given by a trace operator $T$, and the 
fact that the elements of the algebra must preserve all the domains $\Dom {P_T}^k$. Therefore, we have 
to consider a subalgebra of $C^\infty(\ol M)$ that will be adapted to a realization $P_T$:

\begin{definition}
Let $\{P_+,\,T\}$ be an elliptic pseudodifferential boundary system of order one, where 
$P\in \Psi^1(\wt E)$, $T=S\ga_0$ is a normal trace operator, with $\ga_0: u \mapsto u_{\vert N}$ and 
$S$ an endomorphism of $E_N$. Suppose moreover that $P_T$ is selfadjoint (to apply Definition 
\ref{Def}).

We define $\A_{P_T}$ as the $*$-algebra of smooth functions $a\in C^\infty(\ol M)$ such that 
\begin{align*}
&a \,\H^\infty_{P_T} \subseteq  \H^\infty_{P_T} , 
\qquad  a^* \, \H^\infty_{P_T} \subseteq  \H^\infty_{P_T}\, .
\end{align*}
\end{definition} 

\begin{remark}
Note that for any $a\in C^\infty(\ol M)$, we have $a\,\Dom P_T \subseteq \Dom P_T$. As a consequence, 
when $a\in \A_{P_T}$, $[P_T,a]$ is an operator with domain $\Dom P_T$, which sends $\H^\infty_{P_T}$ 
into itself.
\end{remark}

The following lemma provides some lower bounds to $\A_{P_T}$.

\begin{lemma}
\label{aptlem}

(i)  $\A_{P_T}$ contains the algebra 
$$
\B:=\set{a\in C^\infty(\ol M) \, :  \,T\, d^k a \, , \, T\, d^k a^* \in \Psi^\infty(E_N)\,T, \, \text{ for any } k \in \N}
$$
where $d^k :=  [P_+,\cdot \,]^k$.

(ii) If $P$ is a differential operator, $\A_{P_T}$ contains the smooth functions that are constant near the boundary.

\end{lemma}

\begin{proof}
$(i)$
Suppose that $a\in \B$. Since $Ta = a_{|N}\, T$, we directly check that 
$a \Dom P_T \subseteq \Dom P_T$. By induction, for any $j\in \N$,
${[P_+}^j,a]=\sum_{i=1}^j c_{ij} \,d^i(a) \,{P_+}^{j-i}$ where $c_{ij}$ are scalar coefficients. Choose 
$k\in \N$ and $\psi\in \Dom {P_T}^k$. Thus, $\psi \in H^k(E)$ and for any $0\leq j\leq k-1$, 
$T({P_+}^j \psi) =0$. So, there are $R_i \in \Psi^\infty(E_N)$ such that
$$
T ({P_+}^j a \psi) = \sum_{i=1}^j c_{ij} T\big(d^i(a) {P_+}^{j-i}\psi \big) =\sum_{i=1}^j
c_{ij} R_i \,T\,{P_+}^{j-i} \psi =  0
$$
which proves that $a\psi \in \Dom P_T^k$ since $a\psi \in H^k(E)$. The same can be obtained for $a^*$.

$(ii)$ If $a$ is a smooth function constant near the boundary, there is $\la\in \C$ and $f\in C_c^\infty(M)$ with 
compact support such that $a=\la 1_M+f$, where $1_M$ is the function equal to 1 on $M$. The result 
follows from inclusions $C^\infty_c(M) \Dom P_T^k \subseteq H_c^k(E) \subseteq \Dom P_T^k$.
\end{proof}

Here is the main result of this section:

\begin{theorem}
\label{thmboundary}
Let $P\in \Psi^1(\wt E)$ be a symmetric pseudodifferential operator of order one on $\wt M$ satisfying the 
transmission condition.

Let $S \in C^\infty\big(N,\End(E_N)\big)$ be an idempotent selfadjoint endomorphism on the boundary such 
that the system $\set{P_+,\,T:=S\gamma_0}$ is an elliptic pseudodifferential boundary operator.  Then

(i) $P_T$ is selfadjoint if and only if
\begin{align}
\label{Afr}
(1-S) \,\Afr_P\,(1-S)=0 \quad \text{and} \quad S\, \Afr_P^{-1}\, S =0 \, .
\end{align}

(ii) When $P_T$ is selfadjoint, $\big(C^\infty(\ol M), L^2(E), P_T\big)$ is a spectral triple of dimension $d$.

(iii) When $P$ is a differential operator such that $P^2$ has a scalar principal symbol and $P_T$ is 
selfadjoint, the spectral triple $\big(\A_{P_T}, L^2(E), P_T\big)$ is regular.

(iv) Under the hypothesis of $(iii)$, $\A_{P_T}$ is the largest algebra $\A$ in $C^\infty(\ol M)$ such that 
the triple $\big(\A, L^2(E), P_T\big)$ is regular.
\end{theorem}

\begin{proof}
$(i)$  Since $\set{P_+,T}$ is elliptic and $P^*=P$ (viewed as defined on $H^1(\ol M)$), we can apply 
\cite[1.6.11 Theorem]{Grubb} with the same notations, except that $S$ is here not surjective: it is an 
endomorphism only surjective on $E_N^+ := S(E_N)$ with kernel $E_N^-:= (1-S)(E_N)$, so $E_N$ is
the direct orthogonal sum of $E_N^+$ and $E_N^-$ and our $S$ is just replaced by the notation 
$\wh S$ where $\wh  R$ is the surjective morphism associated to the endomorphism $R$ from its domain 
to its range $R(E)$ to avoid confusion. In the notation of \cite{Grubb}, we take $B=P_T$, 
$0=G=K=G'=\wt G=T'$, $\rho=\ga_0$. Thus, $P_T$ is selfadjoint if there is a homeomorphism $\Psi$ from 
$H^s(E_N^+)$ onto $H^s(E_N^-)$, such that (since $P=P^*$ yields $\Afr_P^*=-\Afr_P$)
\begin{equation}
-C'^*\, \Afr_P \,\ga_0 = \Psi \, \wh  S \, \ga_0\, ,\label{eqadjoint}
\end{equation}
with $C'$ satisfying $\wh {(1-S)} C'=Id_{E_N^-}$ and $ C' \wh {(1-S)}=(1-S)$. In other words, $C'$ is the 
injection from $E_N^-$ into $E_N$ and $C'^*=\wh{(1-S)}$.

By \cite[(1.6.52)]{Grubb}, when this is the case, $\Psi$ has the form $\Psi=C'^*\Afr_P^*C$ with 
$\wh{S}C=Id_{E_N^+}$ and $C \wh{S}=S$ (remark that the matrix $I^\times$ is the number 1 here). Note 
that $\Afr_P$ is invertible as a consequence of the ellipticity of $P$. 

Now, suppose that $(1-S) \Afr_P(1-S)=0$ and $S\, \Afr_P^{-1}\, S=0$. We define
$\Psi=-\wh{(1-S)}\,\Afr_P \,C$. This is a homeomorphism from $H^s(E_N^+)$ onto $H^s(E_N^-)$. Indeed, 
if we set $\Psi^{-1} := - \wh {S}\, \Afr_P^{-1} \,C'$, we get
\begin{align*}
\Psi \circ \Psi^{-1} &= \wh{(1-S)} \,\Afr_P \,C \,\wh S \,\Afr_P^{-1}\, C' 
=  \wh {(1-S)}\,\Afr_P \,S \,\Afr_P^{-1} \,C' \\
&= \wh {(1-S)}\, \Afr_P \,\big(S +(1-S)\big) \,\Afr_P^{-1} \,C' = \wh {(1-S)} \, C' = Id_{E_N^-}
\end{align*}
and we also have $\Psi^{-1} \circ \Psi = Id_{E_N^+}$ using $S\, \Afr_P^{-1}\, S=0$.
Moreover,  
$$
\Psi \, \wh S = -\wh{(1-S)} \, \Afr_P \, S = -C'^* \,\Afr_P\, (S +(1-S)) = -C'^* \,\Afr_P\,.
$$
As a consequence, \eqref{eqadjoint} is satisfied and the if part of the
assertion follows.

Conversely, suppose that $P_T$ is selfadjoint. From Green's formula \eqref{green}, we get 
$(\Afr_P\, \ga_0 u , \ga_0 v)_N=0$ for any $u,v\in \Dom P_T$. Since 
$\ga_0 \, : \, H^1(E) \rightarrow H^{1/2}(E_N)$ is surjective, $(\Afr_P\, (1-S)\psi , (1-S) \phi)_N=0$ 
for any $\psi,\phi \in H^{1/2}(E_N)$ and thus $(1-S)\, \Afr_P \,(1-S)=0$. Again, from 
\cite[Theorem 1.6.11]{Grubb} we 
get that $\Psi:= C'^* \,\Afr_P^* \, C$ is a homeomorphism from $H^s(E_N^+)$ onto $H^s(E_N^-)$ and we 
check as before that $\Psi^{-1}:= - \wh {S} \,\Afr_P^{-1}\,C'$ is a right-inverse of $\Psi$, and thus, is the 
inverse of $\Psi$. The equation $\Psi^{-1}\circ \Psi = \Id_{E_N^+}$ yields 
$\wh S\, \Afr_P^{-1} \,S \,\Afr_P \,C=0$, which gives $\wh S\,\Afr_P^{-1} \,S\, \Afr_P \,S=0$. Thus, 
\begin{align*}
\wh S\, \Afr_P^{-1} \,S \,\Afr_P = \wh S\, \Afr_P^{-1} \,S \,\Afr_P \,(1-S) = 
\wh S \,\Afr_P^{-1} \,\big(S +(1-S) \big)\,\Afr_P \,(1-S)= \wh S \,(1-S) =0
\end{align*}
so $S\, \Afr_P^{-1}\, S=0$. 

$(ii)$ Clearly, $C^\infty(\ol M)$ is represented as bounded operators on $L^2(E)$ by left multiplication. Since 
$P_T$ is a selfadjoint unbounded operator on $L^2(E)$, $(P_T-i)^d$ is a bijective operator from 
$\Dom P_T^d$ onto $L^2(E)$.

The system $\set{P_+,\,T}$ being elliptic, it follows from Proposition \ref{propcomposition} that 
$\Dom {P_T}^d$ is a closed subset of $H^d(E)$ and ${P_T}^d$ is continuous from $\Dom {P_T}^d$, with 
the topology of $H^d(E)$, into $L^2(E)$. Using the inverse mapping theorem as in the proof of Lemma
\ref{RegularityLem2} $(ii)$, we see that $(P_T-i)^{-d}$ is a topological isomorphism from $L^2(E)$ onto 
$\Dom {P_T}^d$ (with the induced topology of $H^d(E)$).
Again, \cite[A.4 Lemma]{Grubb} implies that $(P_T-i)^{-d}$ is compact, and the singular values 
$\mu_n(|P_T|)$ are $\mathcal{O}(n^{1/d})$. 

Let $a\in C^\infty(\ol M)$. In particular $\Dom [P_T,a ] = \Dom P_T$ and if $\psi \in \Dom P_T$, 
$[P_T,a] \psi = [P_+,a]\psi$. By Lemma \ref{leftoverlem}, $[P_+,a] = [P,\wt a]_+$ where 
$\wt a \in C^\infty(\wt M)$ is such that $(\wt a)_+=a$.
Thus, since $[P,\wt a]$ is a pseudodifferential operator of order 0 in $(\wt M,\wt E)$ satisfying the 
transmission property, $[P_+,a]$ is continuous from $H^s(E)$ into $H^{s}(E)$ for any $s>-\half$. In 
particular, $[P_T,a]$ extends uniquely as a bounded operator $da$ on $L^2(E)$.

$(iii)$ By $(ii)$ $\big(C^\infty(\ol M), L^2(E), P_T\big)$ is a spectral triple. Thus, since $\A_{P_T}$ is a $*$-subalgebra of $C^\infty(\ol M)$, $\big(\A_{P_T}, L^2(E), P_T\big)$ is also a spectral triple (of the same dimension).

 Let $A \in \A_{P_T} \cup d\A_{P_T}$. It is clear that $A$ sends $\H^\infty_{P_T}$ into itself. We 
denote $B$ the associated $\H^\infty_{P_T}$-bounded operator (so $\ol B=A$).

Clearly, $H^k(E)$ is a Sobolev scale associated to $L^2(E)$ and $P_T$. The result will follow by Lemma 
\ref{RegularityLem2} if we check that for any $k\in \N$, $B^{(k)}:=[P_T^2,\cdot]^k(B)$ 
is continuous from $\H^\infty_{P_T}$ with the $H^k(E)$-topology, into $\H^\infty_{P_T}$ with the 
$L^2(E)$-topology. Since $P$ is differential, Lemma \ref{leftoverlem} yields for any
$\psi\in \H^\infty_{P_T}$ 
$$
B^{(k)}\psi = [P_T^2,\cdot]^k (B)  \psi=[(P_+)^2,\cdot]^k(B) (\psi)= \big(\,[P^2,\cdot]^k(\wt B)\, \big)_+ (\psi) 
$$
where $\wt B$ is a differential operator of order 0 on $(\wt M,\wt E)$ satisfying $\wt B _+ = \ol B=A$. Since 
$P^2$ has scalar principal symbol, $[P^2,\cdot]^k(\wt B)$ is a differential operator of order $k$. The 
claim follows.

$(iv)$ Suppose that $\A$ is a subalgebra of $C^\infty(M)$ such that $(\A,L^2(\ol M,E), P_T)$ is a regular 
spectral triple, $\A$ acting on $L^2(\ol M,E)$ by left multiplication. By regularity, a direct application of the 
proof of $(3)\Rightarrow (4)$ in \cite[Lemma 2.1]{CReconstruction} yields
$a\H^\infty_{P_T} \subseteq \H^\infty_{P_T}$ for any $a\in \A$. This shows that $\A\subseteq \A_{P_T}$.
\end{proof}

Note that if $S=1$ or $S=0$ then by previous theorem $(i)$, $P_T$ is not selfadjoint.

\begin{remark}
The hypothesis ``$P$ is a differential operator'' in Theorem \ref{thmboundary} $(iii)$ is crucial in the sense 
that if $P$ is a non-differential pseudodifferential operator, some non-vanishing leftovers may appear in 
$B^{(k)}$ and destroy the continuity of order $k$. 

This phenomenon does not appear in the boundaryless case since it stems from the singularities 
generated by the cut-off operator $e^+r^+$ at the boundary.
\end{remark}

We conclude this section with a simple one-dimensional example:

\begin{remark}
 Let $\wt M = \mathbb{S}^1$ and $\ol M=\set{(x,y)\in \mathbb S^1 \ : \ x\geq 0}\simeq [-\tfrac{\pi}{2},\tfrac{\pi}{2}]$.  Define $P$ on $H^1(\wt M,\C^2)$ by $ \twobytwo{0}{\tfrac{d}{d\th}}{-\tfrac{d}{d\th}}{0} $, where $\th$ is the polar coordinate $\th \mapsto (\cos \th,\sin\th)$, and let $T=S\ga_0$ where $S=\twobytwo{1}{0}{0}{0}$. The system $(P_+,T)$ is elliptic of order 1 and $P_T$ is selfadjoint since $\Afr_P=\eps \twobytwo{0}{1}{-1}{0}=-\Afr_P^{-1}$ and $(1-S)\Afr_P=\Afr_P S$ where $\eps$ is the function on $N=\del \ol M$ such that $\eps(0,1)=1$ and $\eps(0,-1)=-1$. 
  
  By Theorem \ref{thmboundary}, $(\A_{P_T},L^2(\ol M,\C^2),P_T)$ is a regular spectral triple of dimension 1.
  
 The smooth domain $\H^\infty_{P_T}$ is equal to the set of all $\psi=(\psi_1,\psi_2)\in C^\infty(\ol M,\C^2 )$ such that for all $k\in \N$, $(\psi_1^{(2k)})_{|N}=0$ and $(\psi_2^{(2k+1)})_{|N}=0$, where $\phi^{(p)}$ for all $p\in \N$ denotes $(\tfrac{d}{d\th})^p \phi$.  
 
 Moreover, the algebra $\A_{P_T}$ is equal to the set of all smooth functions $a\in C^\infty(\ol M)$ such that for all $k\in \N$, $(a^{(2k+1)})_{|N}=0$. For example, $\th\mapsto \sin \th \in \A_{P_T}$ while $\th\mapsto \cos\th\notin \A_{P_T}$. In particular $\A_{P_T}$ is strictly included in $C^\infty(\ol M)$.
 
  We conjecture that $\A_{P_T}$ is different from $C^\infty(\ol M)$ for any differential operator of first order $P$ and any normal trace operator $T=S \ga_0$ such that $P_T$ is selfadjoint and $\set{P_+,T}$ is elliptic. 
\end{remark}

\subsection{Case of a Dirac operator}
\label{CasDirac}

We now assume that $\wt M$ is a riemannian manifold with metric $g$, $d=\dim \wt M$ is even and $\wt E$ has 
a Clifford module structure. This means that there is smooth map $\ga : T \wt M \to \End(\wt E)$ such that for any 
$x,y\in T\wt M$, $\ga(x)\ga(y) + \ga(y)\ga(x)=2g(x,y)$. 

We also fix on $\wt E$ a Hermitian inner product $(\cdot,\cdot)$ such that $\ga(x)^*=\ga(x)$ for any $x\in T\wt M$. 
Note that for a given $\ga$, such Hermitian inner product always exists. We fix a connection $\nabla$ on $\wt E$ 
such that for any $x,y \in T\wt M$ and $v,w \in \wt E$,
\begin{align*}
&x(v,w) = (\nabla_x v ,w) + (v,\nabla_x w)\, , \\
&\nabla_x \big(\ga(y) v \big) = \ga \big(\nabla_x^{LC}(y) \big) v + \ga(y)\nabla_x v\, .
\end{align*}
where $\nabla^{LC}$ is the Levi-Civita connection of $(\wt M,g)$. By \cite[Lemma 1.1.7]{Gilkey2}, such connection 
always exists.
The Dirac operator associated to the connection $\nabla$ is locally defined by $\DD := i\sum_j \ga_j \nabla_j$ where  
our conventions are the following: $(e):=\set{e_1,\ldots,e_d}$ is a local orthonormal frame of the 
tangent space where $e_d$ is the inward pointing unit vector field,  
$\ga_j:=\ga(e_j)$ and $\nabla_j:=\nabla_{e_j}$.

The chiral boundary operator $\chi$ investigated in \cite{CC2, Tadpole} is defined the following way: 
choose $\chi:=(-i)^{d/2+1}\ga_1\cdots \ga_{d-1}$, so $\chi_{\wt M}:=i\chi  \ga_d$ (also sometimes denoted 
$\ga_{d+1}$) is the natural chirality of $\wt M$. 
Then $\{\chi,\ga_d\}=0$ while $[\chi,\ga_n]=0$, $\forall n \in \set{1,\cdots,d-1}$. Let $\Pi_{\pm}$ be the 
projections associated to eigenvalues $\pm1$ of $\chi$ ($\chi=\chi^*$ and $\chi^2=1$). 
This defines a particular case of elliptic boundary condition 
\cite[Lemma 1.4.9, Theorem 1.4.11, Lemma 1.5.3]{Gilkey2}. The map $S:=\Pi_-$ is an 
idempotent selfadjoint endomorphism on $N$ and $1-S=\Pi_+$. Since 
$\Afr_{\DD}=-i\ga_d=-\Afr_{\DD}^{-1}$ by Remark \ref{casDirac}, we get 
$(1-S)\,\Afr_{\DD}=\Afr_{\DD}\,S$ and $\DD_T$, where $T:=S_{|N}\ga_0$, is selfadjoint by \eqref{Afr}. This type of 
boundary condition is chosen of course to obtain a selfadjoint boundary Dirac operator, which is not the case if we  
choose a Dirichlet or a Neumann--Robin condition.

In this framework, Theorem \ref{thmboundary} shows the following, a fact not considered in \cite{CC2}:

\begin{theorem}
\label{CasduDirac}
The triple $\big(\A_{\DD_T} , L^2(\ol M,E),\DD_T\big)$ is a regular spectral triple of dimension $d$.
\end{theorem}

\subsection{A spectral triple on the boundary $N:=\del M$}
\label{tripletbord}

We intend here to construct a spectral triple on the boundary of the manifold $\ol M$. The idea is to define 
a transversal differential elliptic operator on the boundary from a differential elliptic operator on
$\wt M$ and a Riemannian structure.
 
Let $ \DD$ be an elliptic symmetric differential operator of order one on $\wt M$. Recall that a pair $(g,h)$ 
where $g$ is a metric on $\wt M$ and $h$ is an hermitian pairing on $\wt E$ is said of the product-type 
near the boundary $N$ if there exists a tubular neighborhood $U$ of $N$ in $\wt M$ such that 
$(U,g_{|U})$ is isometric to $]-\eps,\eps[\times N$ for some $\eps>0$ with metric $dx^2 \otimes g_N(x)$ 
where $g_N(x)$ is a smooth family of metrics on $N$, and $h$ is such that 
$h(x):= \F_* h_{|\set{x} \times N}$ is independent of $x\in ]-\eps,\eps[$, where $\F$ is the bundle 
isomorphism between $\wt E_{|U}$ and $]-\eps,\eps[\times E_N$. 

As observed in \cite[Section 2.1]{BLZ}, we can always suppose that the pair 
$(g,h)$ on $\wt M$ is of the product type near the boundary $N$. The idea is to let the
coefficients of $ \DD$ to absorb any non-product behavior of $(g,h)$. More precisely, if we start with a 
general pair $(g,h)$ on $\wt M$ and if $(g_1,h_1)$ is pair of product type near the boundary, we can 
define $s \in C^\infty\big(\wt M,L(\wt E)\big)$ such that for any $p\in \wt M$, $s(p)$ is the isomorphism 
between the two equivalent quadratic spaces $\big(\wt E_p,h(p)\big)$ and $\big(\wt E_p, h_1(p)\big)$. 
We then define the following application
$$
\Psi : \, u \in C^\infty(\wt M,\wt E) \to \rho^{-1/2}s^{-1} u  \in C^\infty(\wt M,\wt E) 
$$
where $\rho:=d\om(g_1)/d\om(g)$ and $d\om(g)$ (resp. $d\om (g_1)$) is the volume form associated to 
the metric $g$ (resp. $g_1$). This map extends as an isometry between $L^2(\wt M,\wt E,g,h)$ and 
$L^2(\wt M,\wt E,g_1,h_1)$. 
Thus, we can deal with the differential operator $\Psi \, \DD \, \Psi^{-1}$ which is unitarily equivalent to 
$\DD$. 

From now on, we suppose that $(g,h)$ on $(\wt M,\wt E)$ is of product-type.

Thus, $ \DD$, as an operator in $L^2(U,E_{|U})$, is unitarily equivalent (see for instance 
\cite{APS,BL,BLZ}) to an operator of the form 
$$
I_x \circ ( \tfrac{d}{dx} + A_x )
$$
where $I\in C^\infty \big(]-\eps,\eps[,GL(E_N)\big)$, $A\in
C^\infty \big(]-\eps,\eps[,\Diff^1(N,E_N)\big)$, each $A_x$ being elliptic and each $I_x$
being an anti-selfadjoint endomorphism. 
Note that, when $ \DD$ is a Dirac type operator (in the sense that its square 
$ \DD^2$ has $g_x(\xi,\xi)\,\text{Id}$ for principal symbol), $I$ can be chosen as constant $I_x=I_0$ with  
$I_0^2=-1$.

We then define the tangential operator 
$$
\DD_N:= I_0 \circ A_0,
$$
which is an elliptic symmetric first-order differential operator on $(N,E_N)$.

Note that ${\DD_N}^2$ has a scalar principal symbol if $\DD^2$ has.  By Proposition \ref{propmtilde}, we 
directly obtain a spectral triple on the boundary:

\begin{prop}
Let $\DD$ be an elliptic symmetric differential operator of order one on $\wt M$ such that $\DD^2$ has a 
scalar principal symbol. Then $\big(C^\infty(N),L^2(E_N),\DD_N\big)$ is a regular spectral triple of 
dimension $d-1$.
\end{prop}

Remark that if $\DD$ is a classical Dirac operator associated to a Clifford module, $\DD_N$ corresponds 
to the hypersurface Witten--Dirac operator and has been intensively studied in \cite{Ginoux,HMR}.

\section{Reality and tadpoles}
\label{Reality}

\subsection{Conjugation operator and dimension spectrum on a commutative triple}

\begin{definition}
A commutative spectral triple $(\A,\H,\DD)$ provided with a chirality operator $\chi$ (a $\Z/2$ grading on $\H$ 
which anticommutes with $\DD$) is said to be real if there exists an antilinear isometry $J$ on $\H$ such that 
$JaJ^{-1}=a^*$ for $a \in \A$, $J\DD=\epsilon\DD J$, $J\chi=\epsilon' \chi J$, and $J^2=\epsilon''$ where 
$\epsilon, \epsilon', \epsilon''$ are signs given by the table quoted in \cite{ConnesReality,CGravity,Polaris}.
\end{definition}

As we shall see, it turns out that the existence of an operator $J$ only satisfying $J\DD=\pm \DD J$ and 
$J a J^{-1} = a^*$ is enough to impose vanishing tadpoles at any order. We thus introduce a weak 
definition of conjugation operator:

\begin{definition}
\label{charge}
A conjugation operator on a commutative spectral triple $(\A,\H,\DD)$ is an antilinear isometry $J$ on 
$\H$ such that $JaJ^{-1} = a ^*$ for any $a\in \A$ and $J\DD = \eps \DD J$ with $\eps\in \set{-1,1}$.
\end{definition}

In order to be able to compute the spectral action and the corresponding tadpoles on the spectral triples 
$\big(\A_{P_T}, L^2(E), P_T\big)$, we shall need the noncommutative integral and a characterization of 
the dimension spectrum. We first recall a few definitions of the Chamseddine--Connes pseudodifferential
calculus \cite{CC1}. From now on, $(\A,\H,\DD)$ is a regular commutative spectral triple of dimension $d$ 
endowed with conjugation operator $J$.

We shall use the following convention: $X^{-s}$ when $\Re(s)>0$ actually means $(X+P_0)^{-s}$ where $P_0$ is 
the orthogonal projection on the kernel of $X$. Note that in our case, $J \A J^{-1}=\A$, by definition of $J$. 

For any $\a\in \R$, we define $OP^\a:=\set{T\ : \ |\DD|^{-\a}T \in \cap_k \Dom \delta^k}$, where 
$\delta=\ol {[|\DD|,\cdot]}$ as defined in section 2.

\begin{lemma}
\label{Atilde}
Suppose that the triple $(\A,\H,\DD)$ satisfies the first-order condition $[da,b]=0$ for any $a,b\in \A$ and 
let $A$ be a one-form (a finite sum of terms of the form $adb$, for $a,b\in \A$). Then $A^* =- \eps J A J^{-1}$. In 
particular, if $A$ is selfadjoint, $A + \eps J A J^{-1}=0$.
\end{lemma}
\begin{proof}
Direct computation.
\end{proof}

\begin{definition}
\label{defpseudo} Let $\DD(\A)$ be the polynomial algebra generated by $\A$, $\DD$.

The operator $T$ is said to be pseudodifferential if there exists $d\in \Z$ such that for any 
$N\in \N$, there exist $p\in \N$, $P\in \DD(\A)$ and $R\in OP^{-N}$ ($p$, $P$ and $R$ may depend on 
$N$) such that $P\,\DD^{-2p}\in OP^d$ and
$$
T=P\,\DD^{-2p}+R\, .
$$
Define $\Psi(\A)$ as the set of pseudodifferential operators and $\Psi^k(\A):=\Psi(\A)\cap OP^k$.
\end{definition}

\begin{remark}
We use here the algebra of pseudodifferential operators defined by Chamseddine--Connes in 
\cite{CC1} and denoted $\Psi_1(\A)$ in \cite{MCC,Tadpole}. This algebra does not a priori contain the 
operators of type $|\DD|^{k}$ or $|\DD+A|^k$ ($A$ one-form) for $k$ odd, contrarily to the larger algebra 
considered in \cite{MCC}. 
\end{remark}

{\it The dimension spectrum} $Sd(\A,\H,\DD)$ of a spectral triple has initially been defined in
\cite{Cgeom, CM} and is adapted here to the definition of pseudodifferential operator.

\begin{definition}  A spectral triple $(\A,\H,\DD)$ is said to be simple if all generalized zeta functions 
$\zeta_\DD^P := s\mapsto \Tr \big(P |\DD|^{-s}\big)$, where $P$ is any pseudodifferential 
operator in $OP^0$, are meromorphic on the complex plane with only simple poles. 
The set of these poles is denoted $Sd(\A,\H,\DD)$, and called the dimension spectrum of $(\A,\H,\DD)$.
\end{definition}

When $(\A,\H,\DD)$ is simple, $\ncint T := \Res_{s=0} \Tr \big(T |\DD|^{-s}\big)$
is a trace on the algebra $\Psi(\A)$. Note that for $k\in \N$, 
$$
\underset{s=d-k}{\Res} \,\Tr |\DD+A|^{-s} = \ncint K(A)|\DD|^{-(d-k)} = \ncint |\DD+A|^{-(d-k)}
$$ 
where $K(A)$ is a pseudodifferential operator in the sense of Definition \ref{defpseudo} (see 
\cite[Lemma 4.6, Proposition 4.8]{MCC}).  

Suppose now that $(\A,\H,\DD)$ is simple and that $Sd(\A,\H,\DD)\subseteq d-\N$. Moreover, we suppose, as in 
\cite[p. 197]{ConnesMarcolli} that for all selfadjoint one-forms $A$, $\Tr e^{-t (\DD+A)^2}$ has a complete asymptotic 
expansion in real powers of $t$ when $t \downarrow 0$. Thus, 
$S(\DD+A,\Phi,\Lambda)= \Tr \Phi((\DD+A)^2/\Lambda^2)$ satisfies \eqref{formuleaction} where 
$\zeta_{\DD_A}(0)=\zeta_{\DD}(0)+\sum_{q=1}^d \tfrac{(-1)^{q}}{q}\ncint (A \DD^{-1})^q$ (see \cite{CC1}). 

We now obtain from \cite{Tadpole}:
\begin{prop}
\label{propspecaction} Let $A$ be a selfadjoint one-form.
The linear term in $A$ in the $\Lambda^{d-k}$ coefficient of \eqref{formuleaction}, denoted 
$\Tad_{\DD+A}(d-k)$ and called tadpole of order $d-k$, satisfies:
\begin{align*}
&\Tad_{\DD+A}(d-k) =-(d-k)\ncint A \DD |\DD|^{-(d-k) -2}, \quad \forall  k\neq d,\\    
& \Tad_{\DD+ A}(0)=-\ncint  A \DD^{-1}.
\end{align*}
If the triple satisfies the first-order condition $[da,b]=0$ for any $a,b\in \A$, then (see 
\cite[Corollary 3.7]{Tadpole} and Lemma \ref{Atilde}) $\Tad_{\DD+A}(d-k)=0$ for any $k\in \N$.
\end{prop}

\subsection{Tadpoles on $\big(\A_{\DD_T} , L^2(\ol M,E),\DD_T\big)$}

It is known that real the commutative spectral triple based on the classical Dirac operator on a compact spin 
manifold of even dimension $d$ has no tadpoles \cite{Tadpole}. In fact, as we shall see, the same result applies in 
presence of a boundary. 

We now work in the setting of section \ref{CasDirac}. We suppose moreover that $\wt M$ is a spin manifold, $\wt E$ 
is the spin bundle and $\DD$ is the classical Dirac operator in the sense of Atiyah--Singer.
The spin structure brings us an antilinear isometry $J$ (the ordinary conjugation operator) 
satisfying $\DD J = J\DD$,
$JbJ^{-1}=b^*$ for all $b \in C^\infty(\wt M)$, and $J \ga_i = - \ga_i J$ for any $i\in \set{1,\cdots, d}$ 
\cite[Theorem 9.20]{Polaris}. Moreover, $J \chi_{\wt M} = \eps' \chi_{\wt M} J$, and thus $J\chi = \eps' \chi J$, where $\eps'=-1$ if $d/2$ is odd and $\eps'=1$ if $d/2$ is even. 

In particular, if $d/2$ is even then $J S= SJ$. 

If we set
$$
\wt J:= J\chi_{\wt M}
$$ 
then $\wt J$ is an antilinear isometry satisfying $\DD \wt J = - \wt J \DD$, $\wt J b \wt J^{-1} = b^*$ for all $b\in C^\infty(\wt M)$, and $\wt J \chi = -\eps' \chi \wt J$. As a consequence, if $d/2$ is odd, then $\wt JS = S \wt J$.

\begin{theorem}
\label{notad}
The spectral triple $\big(\A_{\DD_T} , L^2(\ol M,E),\DD_T\big)$ of Theorem \ref{CasduDirac} has no 
tadpoles.
\end{theorem}

\begin{proof}
By Theorem \ref{CasduDirac}, we know that $\big(\A_{\DD_T} , L^2(\ol M,E),\DD_T\big)$ is a regular 
spectral triple of dimension $d$. In order to prove it is a simple spectral triple with dimension spectrum included
$\set{d-k \ : \ k\in \N}$, it is enough to check that for any $B\in \DD(\A_{\DD_T})$, the function
$\zeta^B_{\DD_T}(s):=\Tr \big(B |\DD_T|^{-s}\big)$ has only simple poles in $\Z$. Since
$\DD$ is differential, by Lemma \ref{leftoverlem}, $B$ is a differential operator on $(\ol M,E)$. Since 
$s\mapsto \Tr \big(A |\DD_T|^{-s}\big)$ has a meromorphic extension on $\C$ with only 
simple poles in $\set{d+n-k \ : k\in \N}$, when $A$ is a differential operator of order $n$ on $(\ol M,E)$ 
(see for instance \cite[Theorem 1.12.2]{Gilkey}), we can conclude that $\big(\A_{\DD_T} , L^2(\ol M,E),\DD_T\big)$ 
is simple with dimension spectrum included in $d-\N$. Moreover, if $A$ is a selfadjoint one-form, 
$\Tr e^{-t (\DD_T+A)^2}$ has a complete asymptotic expansion in real powers of $t$ when $t\to 0$ 
\cite[Theorems 1.4.5 and 1.4.11]{Gilkey2}.

Note that the first order condition $[da,b]=0$ for any $a,b\in \A_{\DD_T}$ is clearly satisfied since $da$ is a 
differential operator of order 0 on $(\ol M,E)$.

It only remains to prove that there is conjugation operator (Definition \ref{charge}) on the simple spectral triple 
$\big(\A_{\DD_T} , L^2(\ol M,E),\DD_T\big)$. 
We define $J':= J$ if $d/2$ is even, and $J':=\wt J$ if $d/2$ is odd.

The operator $J'_+:= r^+ J' e^+$, is an endomorphism on 
$L^2(\ol M,E)$. Clearly, $J'_+$ is an antilinear isometry satisfying $J'_+ a {J'_+}^{-1} = a^*$
for any $a\in \A_{\DD_T}$. By Lemma \ref{leftoverlem}, $L(J',\DD)=L(\DD,J')=0$ and
thus, $J'_+ \DD_+ = (J'\DD)_+ = (-1)^{d/2}(\DD J')_+ = (-1)^{d/2}\DD_+ J'_+$. 

Moreover $J'_{|N}\, S_{|N}\,{J'^{-1}}_{|N}=S_{|N}$ and thus   
$$
T J'_+ = S_{|N} \ga_0 J'_+ = S_{|N} J'_{|N} \ga_0 = J'_{|N}  S_{|N}\ga_0 = J'_{|N} T.
$$
In particular, $J'_+$ preserves $\Dom \, \DD_T$, and thus $J'_+ \DD_T = (-1)^{d/2}\DD_T J'_+$. As a consequence, 
$J'_+$ is a conjugation operator on 
$\big(\A_{\DD_T} , L^2(\ol M,E),\DD_T\big)$. 
Proposition \ref{propspecaction} now yields the result.
\end{proof}

\section*{Acknowledgments}

\hspace{\parindent}
We thank Alain Connes, Ali Chamseddine, Gerd Grubb, Ryszard Nest and Uuye Otgonbayar for helpful 
discussions.


\begin{thebibliography}{60}

\bibitem{ANS}
J. Aastrup, R. Nest and E. Schrohe, 
``Index theory for boundary value problems via continuous fields of C$^*$-algebras'', 
Journal of Functional Analysis {\bf 257} (2009), 2645--2692.

 \bibitem{APS}
 M. Atiyah, V. Patodi and I. M. Singer, 
 ``Spectral asymmetry and Riemannian geometry, 
 
 I. Math. Proc. Cambr. Phil. Soc {\bf 77} (1975), 43--69,
 
 II. Math. Proc. Cambr. Phil. Soc {\bf 78} (1975), 405--432,
 
 III. Math. Proc. Cambr. Phil. Soc {\bf 79} (1976), 71-99.

\bibitem{BB}
B. Boo{\ss}-Bavnbek,
``The determinant of elliptic boundary problems for Dirac operators",
  
http://mmf.ruc.dk/~booss/ell/determinant.pdf

\bibitem{BBL}
B. Boo{\ss}-Bavnbek and M. Lesch,
``The invertible double of elliptic operators", 
Lett. Math. Phys. {\bf 87} (2009), 19--46.

\bibitem{BLZ}
B. Boo{\ss}-Bavnbek, M. Lesch and C. Zhu, ``The Calder\'on projection: new
definitions and applications'', Journal of Geometry and Physics \textbf{59}
(2009), 784--826.

\bibitem{BW}
B. Boo{\ss}-Bavnbek and K. Wojciechowski,
\emph{Elliptic Boundary Problems for Dirac Operators},
Birkh\"auser, Boston, 1993.

\bibitem{BG1}
T. Branson and P. Gilkey,
``The asymptotics of the Laplacian on a manifold with boundary",
Partial Differential Equations \textbf{15} (1990), 245--272.

\bibitem{BG2}
T. Branson and P. Gilkey,
``Residues of the eta function for an operator of Dirac type with local boundary conditions", 
Differential Geometry and its Applications \textbf{2} (1992), 249--267.

\bibitem{BL}
J. Br\"{u}nig and M. Lesch, ``On boundary value problems for Dirac type
operators'', Journal of Functional Analysis \textbf{185} (2001), 1--62.

\bibitem{CKW}
A. L. Carey, S. Klimek and K. P. Wojciechowski,
``A Dirac type operator on the non-commutative disk'', Lett. Math. Phys. 
\textbf{93} (2010), 107--125.

\bibitem{CC}
A. Chamseddine and A. Connes,
``The spectral action principle",
Commun. Math. Phys. {\bf 186} (1997), 731--750.

\bibitem{CC1}
A. Chamseddine and A. Connes,
``Inner fluctuations of the spectral action",
J. Geom. Phys. {\bf 57} (2006), 1--21.

\bibitem{CC2}
A. Chamseddine and A. Connes,
``Quantum gravity boundary terms from the spectral action on
noncommutative space",
PRL {\bf 99} (2007), 071302.

\bibitem{Book}
A. Connes,
\emph{Noncommutative Geometry},
Academic Press, London and San Diego, 1994.

\bibitem{Cgeom}
A. Connes,
``Geometry from the spectral point of view",
Lett. Math. Phys. {\bf 34} (1995), 203--238.

\bibitem{ConnesReality}
A. Connes,
``Noncommutative geometry and reality",
J. Math. Phys. {\bf 36} (1995), 6194--6231.

\bibitem{CGravity}
A. Connes,
``Gravity coupled with matter and the foundation of non commutative
geometry",
Commun. Math. Phys. {\bf 182} (1996), 155--177.

\bibitem{Connesaction}
A. Connes,
``The action functional in noncommutative geometry",
Comm. Math. Phys. {\bf 117} (1998), 673--683.

\bibitem{CReconstruction}
A. Connes,
``On the spectral characterization of manifolds",
arXiv:0810.2088v1.

\bibitem{Connesbord}
A. Connes, 
``Variation sur le th\`eme spectral'', R\'esum\'e des cours 2006--2007,

http://www.college-de-france.fr/media/ana\_geo/UPL53971\_2.pdf, 2007.

\bibitem{ConnesMarcolli}
A. Connes and M. Marcolli,
\emph{Noncommutative Geometry, Quantum Fields and Motives}, 
Colloquium Publications, Vol. 55, American Mathematical Society, 2008.

\bibitem{CM}
A. Connes and H. Moscovici,
``The local index formula in noncommutative geometry",
Geom. Funct. Anal.  {\bf 5}  (1995), 174--243.

\bibitem{Dix}
J. Dixmier, 
``Existence de traces non normales'', 
C. R. Acad. Sci. Paris, S\' er. A {\bf 262} (1966), 1107--1108.

\bibitem{MCC}
D. Essouabri, B. Iochum, C. Levy and A. Sitarz,
``Spectral action on noncommutative torus", J. Noncommut.
Geom. {\bf 2} (2008), 53--123.

\bibitem{FGLS}
B. Fedosov, F. Golse, E. Leichtnam and E. Schrohe,
``The noncommutative residue for manifolds with boundary",
J. Funct. Anal. {\bf 142} (1996), 1--31.

\bibitem{Gilkey}
P.~B.~Gilkey, {\em Invariance Theory, the Heat equation, and the
Atiyah--Singer Index Theory}, CRC Press, Boca Raton, 1995.

\bibitem{Gilkey2}
P. B. Gilkey, {\it Asymptotic Formulae in Spectral Geometry}, Chapman
\& Hall, 2003.

\bibitem{Ginoux}
N. Ginoux,
\emph{The Dirac spectrum}, Lecture Notes in Mathematics, Vol. 1976, 2009.

\bibitem{Polaris}
J. M. Gracia-Bond\'{\i}a, J. C. V\'arilly and H. Figueroa,
\emph{Elements of Noncommutative Geometry},
Birkh\"auser Advanced Texts, Birkh\"auser, Boston, 2001.

\bibitem{Grubb}
G. Grubb,
\emph{Functional calculus of pseudodifferential boundary problems}, Second
edition, Progress in Mathematics \textbf{65}, Birkh\"auser, Boston, 1996.

\bibitem{Grubb1}
G. Grubb, 
\emph{Distributions and Operators}, Graduate Texts in Math. {\bf 252}, Springer, 2009.

\bibitem{GSc}
G. Grubb and E. Schrohe,
``Trace expansions and the noncommutative residue for manifolds with
boundary,
J. Reine Angew. Math. (Crelle's Journal) {\bf 536} (2001), 167--207.

\bibitem{GSc1}
G. Grubb and E. Schrohe, 
``Traces and quasi-traces on the Boutet de Monvel algebra``, 
Ann. Inst. Fourier {\bf 54} (2004), 1641--1696. 

\bibitem{HH}
S. Hawking and G. Horowitz, 
``The gravitational Hamiltonian, action, entropy and surface terms'', 
Class. Quantum Grav. {\bf 13} (1996), 1487--1498.

\bibitem{Higson1}
N. Higson,
``The local index formula in noncommutative geometry",
Lectures given at the School and Conference on Algebraic K-theory and its applications, Trieste, 2002.

\bibitem{Higson2}
N. Higson,
``The residue index theorem of Connes and Moscovici'', 
Surveys in noncommutative geometry, Clay Math. Proc., vol. 6, Amer. Math. Soc., Providence, RI, 2006, 
71--126.

\bibitem{HMR}
O. Hijazi, S. Montiel and A. Rold\' an,
``Eigenvalue boundary problems for the Dirac operator",
Commun. Math. Phys. {\bf 231} (2002), 375--390.

\bibitem{Hormander}
L. H\"ormander,
\emph{The Analysis of Linear Partial Differential Operators I, II, III},
Springer-Verlag, Berlin, 1989, 2005, 2007.

\bibitem{Tadpole}
B. Iochum and C. Levy,
``Tadpoles and commutative spectral triples",
 arXiv:0904.0222 [math-ph], to appear in J. Noncommut. Geom.

\bibitem{KW}
W. Kalau and M. Walze,
``Gravity, non-commutative geometry, and the Wodzicki residue",
J. Geom. Phys. {\bf 16} (1995), 327--344.

\bibitem{Kastler}
D. Kastler,
`` The Dirac operator and gravitation",
Comm. Math. Phys. {\bf 166} (1995), 633--643.

\bibitem{KV}
M. Kontsevich and S. Vishik, 
``Geometry of determinants of elliptic operators'', 
Functional Analysis on the Eve of the 21'st century, Vol. I (S. Gindikin et al. eds.), Prog. Math. 131, 
Birkh\"auser, Boston, 1995, 173--197.

\bibitem{Lawson}
H. B. Lawson and M.-L. Michelsohn,
\emph{Spin Geometry}, Princeton Univ. Press, Princeton, 1989.

\bibitem{Lesch}
M. Lesch,
\emph{Operators of Fuchs type, conical singularities and asymptotic methods}, 
Teubner, Leipzig, 1997.

\bibitem{Lesch1}
M. Lesch,
``On the noncommutative residue for pseudodifferential operators with
log-polyhomogeneous symbols",
Ann. Global Anal. Geom. {\bf 17} (1999), 151--187.

\bibitem{Lescure}
J.-M. Lescure,
``Triplets spectraux pour les vari\'et\'es \`a singularit\'e conique
isol\'ee",
Bull. Soc. Math. France {\bf 129} (2001), 593--623.

\bibitem{LM}
J.L. Lions and E. Magenes, \emph{Probl\`{e}mes aux limites non homog\`{e}nes et
applications, Volume 1}, Dunod, Paris, 1968.

\bibitem{NS}
R. Nest and E. Schrohe, 
``Dixmier's trace for boundary value problems'', 
manuscripta math. {\bf 96} (1998), 203--218.

\bibitem{Otogo}
U. Otgonbayar, 
``Pseudo-differential operators and regularity of spectral triples'',

arXiv:0911.0816 [math.OA].

\bibitem{Rennie}
A. Rennie, 
``Smoothness and locality for nonunital spectral triples'', 
K-Theory {\bf 28} (2003), 127--165.

\bibitem{Schrohe1}
E. Schrohe,
``Noncommutative residues and manifolds with conical singularities",
J. Funct. Anal. {\bf 150} (1997), 146--174.

\bibitem{Schrohe}
E. Schrohe,
``Noncommutative residues, Dixmier's trace, and heat trace expansions
on manifolds with boundary",
In: B. Booss-Bavnbek and K. Wojciechowski (eds), Geometric Aspects of
Partial Differential Equations. Contemporary Mathematics, {\bf 242}
Amer. Math. Soc. Providence, R.I., (1999), 161--186.

\bibitem{Uga}
W. J. Ugalde, 
``Some conformal invariants from the noncommutative residue for manifolds with boundary'',
SIGMA, {\bf 3} (2007), 104, 18 pages.

\bibitem{Vassi}
D. V. Vassilevich,
``Heat kernel expansion: user's manual,
Phys. Rep. {\bf 388} (2003), 279--360.

\bibitem{Wang1}
Y. Wang,
``Differential forms and the Wodzicki residue for manifold with boundary'', 
J. Geom. Physics, {\bf 56} (2006), 731--753.

\bibitem{Wang2}
Y. Wang,
``Differential forms and the noncommutative residue for manifolds with boundary in the non-product case'' 
Lett. Math. Phys. (2006), 41--51.

\bibitem{Wang3}
Y. Wang,
``Gravity and the noncommutative residue for manifolds with boundary,
Lett. Math. Phys. {\bf 80} (2007), 37--56.

\bibitem{Wodzicki1}
M. Wodzicki,
``Local invariants of spectral asymmetry",
Invent. Math. {\bf 75} (1984), 143--177.

\bibitem{Wodzicki3}
M. Wodzicki, ``Noncommutative residue. Chapter I: Fundamentals",
320--399, in {\em K-theory, Arithmetic and Geometry}, Yu. I. Manin,
ed., Lecture Notes in Mathematics {\bf 1289}, Springer, Berlin 1987.

\bibitem{Yang}
C. Yang,
``Isospectral deformations of Eguchi--Hanson spaces as non unital spectral triples",
Commun. Math. Phys. {\bf 288} (2009), 615--652.


\end{thebibliography}
\end{document}